\documentclass[12pt,draftclsnofoot,journal,onecolumn]{IEEEtran}
\IEEEoverridecommandlockouts
\usepackage{mathtools}
\usepackage{enumitem} 
\usepackage{ifpdf}
\usepackage{amsthm}
\usepackage{amsmath}
\usepackage{nicefrac}
\usepackage{cite}
\usepackage{amsfonts}
\usepackage{comment}
 \usepackage{url}
 \usepackage{amssymb}
 \usepackage[paper=letterpaper,margin=0.75in]{geometry}
 \usepackage[ansinew]{inputenc}

\usepackage{lipsum}
\usepackage{hyperref}

\newtheorem{thm}{Theorem}
\theoremstyle{definition}

\newtheorem{fact}{Fact}
\newcommand*{\QEDB}{\hfill\ensuremath{\square}}%

\makeatletter
\renewenvironment{proof}[1][\proofname] {\par\pushQED{\qed}\normalfont\topsep6\p@\@plus6\p@\relax\trivlist\item[\hskip\labelsep\bfseries#1\@addpunct{.}]\ignorespaces}{\popQED\endtrivlist\@endpefalse}
\makeatother

\begin{document}
\title{Towards Provably Invisible Network Flow Fingerprints}

\author{
	\IEEEauthorblockN{Ramin Soltani\IEEEauthorrefmark{1},
		Dennis Goeckel\IEEEauthorrefmark{1}, Don Towsley\IEEEauthorrefmark{2}, and Amir Houmansadr\IEEEauthorrefmark{2}}
	
	\IEEEauthorblockA{\IEEEauthorrefmark{1}Electrical~and~Computer~Engineering~Department,~University~of~Massachusetts,~Amherst,
		\{soltani, goeckel\}@ecs.umass.edu\\}
	\IEEEauthorblockA{\IEEEauthorrefmark{2}College of Information and Computer Sciences, University of Massachusetts, Amherst,
		\{towsley, amir\}@cs.umass.edu}
	    
                       \thanks{ This work has been supported by the National Science Foundation under grants
                       CNS-1564067 and  CNS-1525642.}
                   \thanks{ This work has been presented at the 51st Annual Asilomar Conference on Signals, Systems, and Computers,  November 2017.}
                   \thanks{Personal use of this material is permitted. Permission from IEEE must be obtained for all other uses, in any current or future media, including reprinting/republishing this material for advertising or promotional purposes, creating new collective works, for resale or redistribution to servers or lists, or reuse of any copyrighted component of this work in other works. DOI: 10.1109/ACSSC.2017.8335179}

}

\date{}
\maketitle
\thispagestyle{plain}
\pagestyle{plain}

\newtheorem{definition}{Definition}

\begin{abstract}
Network traffic analysis reveals important information even when messages are encrypted. We consider active traffic analysis via flow fingerprinting by invisibly embedding information into packet timings of flows. In particular, assume Alice wishes to embed fingerprints into flows of a set of network input links, whose packet timings are modeled by Poisson processes, without being detected by a watchful adversary Willie. Bob, who receives the set of fingerprinted flows after they pass through the network modeled as a collection of independent and parallel $M/M/1$ queues, wishes to extract Alice's embedded fingerprints to infer the connection between input and output links of the network. 
We consider two scenarios: 1) Alice embeds fingerprints in all of the flows; 2) Alice embeds fingerprints in each flow independently with probability $p$. Assuming that the flow rates are equal, we calculate the maximum number of flows in which Alice can invisibly embed fingerprints while having those fingerprints successfully decoded by Bob. Then, we extend the construction and analysis to the case where flow rates are distinct, and discuss the extension of the network model.
\end{abstract}

\section{Introduction}
Security in computer networks has emerged as an important area of research. Although encryption hides information sent from a transmitter, network traffic analysis can extract important information from the size, count, and timings of the packets. For instance, when attackers relay their flows through compromised nodes called stepping stones, traffic analysis can trace back the attackers~\cite{staniford1995holding,zhang2000detecting}. Also, traffic analysis can find the correlations in traffic patterns to link incoming/outgoing flows and break anonymity~\cite{syverson2001towards}.

Flow watermarking and flow fingerprinting are two active traffic analysis methods that work by perturbing packet timings of flows according to specific patterns to embed information in them. In flow watermarking, the information embedded in a flow is one bit, i.e., either the flow is marked or not. However, in many applications, more than one bit of information is required to be embedded in the packet timings of the flows. Flow fingerprinting provides the solution for such applications by embedding several bits of information in the flows such as the information about the party that has embedded the fingerprint, the source of the flow, or the location at which the flow has been fingerprinted~\cite{houmansadr2012design}.

Active traffic analysis has become an important area of research due to the increasing use of encryption. Wang et al.~\cite{wang2003robust} proposed to embed flow watermarks in inter-packet delays to detect stepping stones, and Wang et al.~\cite{wang2005tracking} used an interval-based flow watermark to compromise anonymized VoIP conversations. Houmansadr et al. proposed the first non-blind watermark, RAINBOW~\cite{houmansadr2009rainbow}, offering significantly higher invisibility compared to prior designs, and  SWIRL~\cite{houmansadr2011swirl} was designed to resist aggregated-flows attacks. Houmansadr et al.~\cite{houmansadr2013need} was the first to introduce flow fingerprinting, and TagIt~\cite{rezaei2017tagit} introduced the first blind flow fingerprint. 

\begin{figure}
	\begin{center}
		\includegraphics[width=\textwidth/2 ,height=\textheight,keepaspectratio]{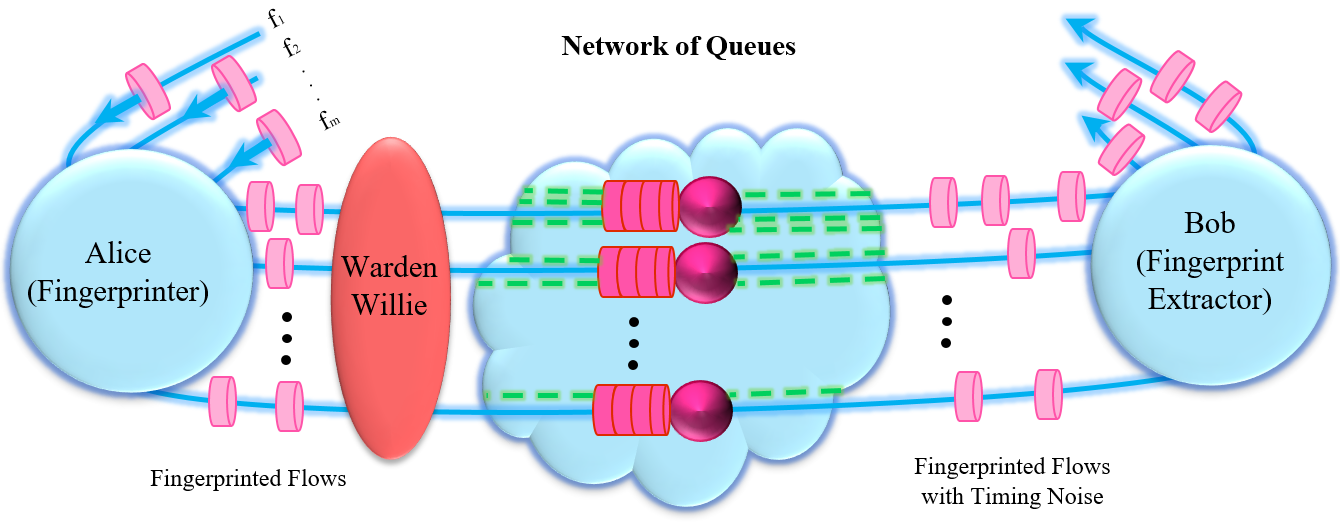}
	\end{center}
	\caption{Alice may embed fingerprints in flows. Bob receives the potentially fingerprinted flow after it passes through a network of $M/M/1$ queues, which are independent and parallel. Each queue is shared between a fingerprinted flow (shown by a blue solid line) and interfering flows (shown by green dotted lines).} 
		\label{fig:SysMod}
	\end{figure}

As previous active traffic analysis designs are based on ad hoc heuristics (such as moving packets into secret time intervals), they do not offer any theoretical guarantees on the invisibility-performance trade-off. In this work, we take a systematic approach to design a flow fingerprinting system with provable information-theoretic guarantees on invisibility and performance (e.g., number of fingerprints). Consider a network containing $m$ independent, parallel, work conserving, and First In First Out (FIFO) queues with independent exponential service times where the $i^{\mathrm{th}}$ queue conveys the $i^{\mathrm{th}}$ flow ($f_i$) from the $i^{\mathrm{th}}$ input link to the $i^{\mathrm{th}}$ output link, and conveys interfering flows independent of $f_i$ (See Fig.~\ref{fig:SysMod}). The network is anonymous to Alice and Bob such that they do not know the connections between input and output links. Alice has access to the input links and is able to buffer the packets and release them when she desires to embed fingerprints in packet timings of the flows. Adversary Willie is between Alice and the network and observes $f_1,f_2,\ldots,f_m$ after they are accessed by Alice and wishes to detect if Alice is embedding fingerprints in the flows or not. Bob observes the packet timings of the flows on the output links and wishes to extract Alice's fingerprints.

We consider the following problem: given the time interval $[0,T]$, can Alice embed flow identifier fingerprints invisible to Willie in the packet timings such that Bob can extract them successfully to de-anonymize the network and, if yes, what is the maximum number $m$ of flows that can be fingerprinted? For the case where the packet timings of the flows are governed by Poisson process, we calculate the asymptotic expression for the number of flows that can be fingerprinted as a function of $T$, for two scenarios: 1) Alice embeds fingerprints in all of the flows; 2) Alice embeds fingerprints in each flow independently with some small probability $p$.

The remainder of the paper is organized as follows. In
Section~\ref{prerequisites}, we present the system model, definitions and the metrics employed for the two scenarios of interest. Then, we provide constructions and analyses for the two fingerprinting scenarios in Sections~\ref{scen1} and Sections~\ref{scen2}. In Section~\ref{disc} we discuss the results, the extension of the scenarios to distinct flow rates, and the extension of the network model to more general networks. Finally, we conclude in Section~\ref{sec:con}.
\section{System Model, Definitions, And Metrics} \label{prerequisites}

\subsection{System Model} \label{sec:sysmod}
Alice has access to a set of input links  $L_1^{(I)},L_2^{(I)},\ldots,L_m^{(I)}$ of a network, and is able to buffer packets and release them when she desires. The packet flow conveyed over $L_i^{(I)}$ is denoted by $f_i$ ($1\leq i \leq m$), and $\mathcal{F}=\{f_1,f_2,\ldots,f_{m}\}$ is the set of flows accessed by Alice. Bob receives the flows $f_1, f_2,\ldots,f_m$ from the output links $L_1^{(O)},L_2^{(O)},\ldots,L_m^{(O)}$ of the network, respectively. The network is anonymous such that Alice and Bob do not know the connections between input and output links; they wish to infer this in the interval $[0,T]$, and thus de-anonymize the network.

Alice embeds a unique flow identifier fingerprint in each flow by altering its packet timings according to a secret codebook of fingerprints shared with Bob, and Bob extracts the fingerprints from the observed flows.  Warden Willie, who is between Alice and the network, observes the input links and wishes to detect if Alice embeds fingerprints in them or not (see Fig.~\ref{fig:SysMod}). Willie knows the fingerprinting scheme Alice will employ if she chooses to embed fingerprints, but he does not have access to the codebook of fingerprints.

Alice, Bob, and Willie know that the packet timings of the flows $f_1,f_2,\ldots, f_m$ are modeled by Poisson processes with rates $\lambda_1,\ldots,\lambda_m$, respectively. The network consists of $m$ independent single server queues with exponential service times, i.e., $M/M/1$ queues, which are work conserving and First In First Out (FIFO) discipline. Each queue has multiple inputs and outputs such that the $i^{\mathrm{th}}$ queue ($q_i$) conveys $f_i$ from the input link $L_{i}^{(I)}$ to the output link $L_{i}^{(O)}$, and also conveys interfering Poisson flows independent of $f_i$. We denote the sum of the rates of the interfering flows on $q_i$ by $\lambda'_i$. The service rate of $q_i$ is $\mu_i$, and the queues are stable, i.e., $\lambda_i + \lambda'_i\leq \mu_i$.

We consider two scenarios. In Scenario~1 (analyzed in Section~\ref{scen1}), the flow rates are equal ($\lambda_i=\lambda$), and Alice embeds fingerprints in all of the flows of $\mathcal{F}$. In Scenario~2 (analyzed in Section~\ref{scen2}), the flow rates are equal, but Alice embeds fingerprints in each flow independently with probability $p$. For each scenario, we calculate the number of flows in which Alice can invisibly and reliably embed fingerprints, as described precisely next.

\subsection{Definitions}
Willie's hypotheses are $H_0$ (Alice did not embed fingerprints) and $H_1$ (Alice embedded fingerprints). We denote by $\mathbb{P}_{\mathrm {FA}}$ the probability of rejecting $H_0$ when it is true (type I error or false alarm), and $\mathbb{P}_{\mathrm {MD}}$ the probability of rejecting $H_1$ when it is true (type II error or mis-detection). We assume that Willie uses classical hypothesis testing with equal prior probabilities and seeks to minimize his probability of error, $\mathbb{P}_{\mathrm e}^{(\mathrm w)}=\frac{\mathbb{P}_{\mathrm {FA}} + \mathbb{P}_{\mathrm {MD}}}{2}$; the generalization to arbitrarily prior probabilities is available in~\cite{bash_jsac2013}. 

\begin{definition} (Invisibility) Alice's fingerprinting is {\em invisible} (covert) if and only if she can lower bound Willies' probability of error ($\mathbb{P}_{\mathrm e}^{(\mathrm w)}$) by $\frac{1}{2}-\epsilon$ for any $\epsilon>0$, asymptotically. This definition is similar to that of covertness developed in~\cite{bash_jsac2013}, and used in covert communication~\cite{soltani2014covert,soltani2015covert,soltani2016allerton,soltani2017covert}

\end{definition}
\begin{definition} (Reliability) Fingerprinting for each flow is {\em reliable} if and only if $\mathbb{P}_{\mathrm{f}}\leq \zeta$ for any $\zeta>0$, where $\mathbb{P}_{\mathrm{f}}$ is the probability of the failure event which occurs when
\begin{itemize}
\item Alice cannot successfully embed a fingerprint since she does not have a packet to release when she needs to; 
\item Alice runs out of fingerprints because the number of fingerprints in her codebook is less than the number of flows in which she wishes to embed fingerprints; or
\item Bob cannot extract a fingerprint  successfully.
\end{itemize}
\end{definition}

\begin{definition} \label{def:3} (Lambert-W function) The Lambert-W function is the inverse function of $f(W)=We^W$.
\end{definition}

\begin{definition} (The Kullback–Leibler divergence) If $f$ and $g$ are probability measures over a set $\mathcal{S}$, then the Kullback–Leibler divergence between $f$ and $g$ is:
\begin{align}
\mathcal{D}(f||g) = \int_{\mathcal{S}} f(x) \log \frac{f(x)}{g(x)} dx
\end{align}

\end{definition}
	 
In this paper, we use standard Big-O notation~\cite[ch. 3]{cormen2009introduction}.

\section{Scenario~1: All flows are fingerprinted} \label{scen1}
In this section we consider Scenario~1: Alice embeds fingerprints in all of the input flows of a network with equal rates in the time interval $[0,T]$, and Bob extracts the fingerprints from the output flows to infer the connection between input and output flows of the network. Because Alice is able to buffer packets and release them when she desires, she changes the packet timings of the flows to embed fingerprints in them according to a secret fingerprinting codebook shared with Bob. Each of the fingerprints is a flow identifier and consists of a sequence of inter-packet delays to be employed to embed the corresponding fingerprint. To successfully change the packet timings of a flow according to the chosen codeword, Alice must have a packet in her buffer to transmit at the appropriate times. To account for this, Alice uses a two phase scheme for each flow $f_i$, similar to the one adopted in~\cite[Section IV]{soltani2015covert}.  First, Alice slows down $f_i$ to buffer packets; then, during the fingerprinting phase, she releases the packets from her buffer with the inter-packet delays prescribed by the codeword corresponding to the fingerprint, while buffering the arriving packets of $f_i$.  

We calculate the asymptotic expression for the number of flows that can be fingerprinted as a function of $T$ using this strategy.
\begin{thm} \label{thm:allflows} 
	Consider the setting in Section~\ref{sec:sysmod}. In a set $\mathcal{F}$ containing $m$ flows with rates $\lambda_i=\lambda$, Alice and Bob can invisibly and reliably track all of the flows in the time interval $[0,T]$, as long as $m=\mathcal{O}(T/W(T))$, where $W(\cdot)$ is the Lambert-W function.
\end{thm}

\begin{proof}
\textbf{Construction}: 
Per above, Alice divides the time interval of length $T$ into two phases: a buffering phase of length $T_1$ and a fingerprinting phase of length $T_2$ such that $T=T_1+T_2$. During the buffering phase, Alice slows down the packets of each flow $f_i$, from rate $\lambda$ to rate $\lambda-\Delta$, in order to build up packets in her buffer, ensuring that with high probability, she will not run out of packets during the fingerprinting phase of length $T_2=T-T_1$ (see Fig.~\ref{fig:Twophased}).
\begin{figure}
\begin{center}
\includegraphics[width=\textwidth/2,height=\textheight,keepaspectratio]{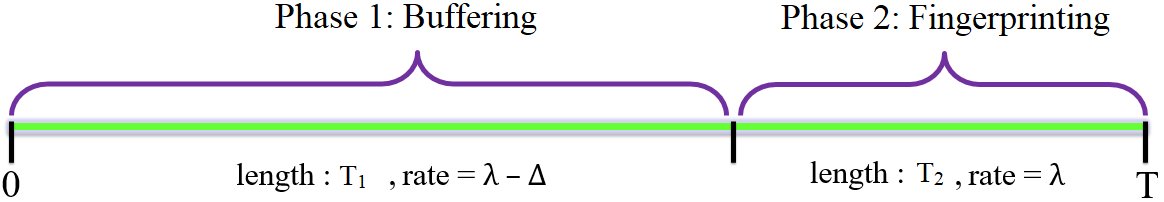}
\end{center}
 \caption{Two phase scheme: Alice divides the duration of time $T$ into two phases with lengths $T_1$ and $T_2=T-T_1$. In the first phase, Alice slows down each flow from the rate $\lambda$ to the rate $\lambda-\Delta$, and buffers the excess packets. In the next phase, she transmits packets at rate $\lambda$ according to the inter-packet delays in the codeword corresponding to the fingerprint to be embedded.}
 \label{fig:Twophased}
 \end{figure}

Alice and Bob share a codebook to which Willie does not have access.The codebook construction is similar to that of~\cite{soltani2015covert,soltani2016allerton}. To build the codebook, a set of $m$ codewords $\{C(W_l)\}_{l=1}^{l=m}$ are independently generated according to realizations of a Poisson process with parameter $\lambda$. In particular, to generate a codeword $C(W_l)$, first a random variable $N$ is generated according to a Poisson distribution with mean $\lambda T_2$. Then, $N$ inter-packet delays are generated by placing $N$ points uniformly and independently on an interval of length $T_2$ \cite{verdubitsq} (see  Fig.~\ref{fig:codebook}). Therefore, each codeword of the codebook is a series of inter-packet delays and corresponds to a unique flow identifier fingerprint. To embed a fingerprint in a flow $f_i$, Alice applies the inter-packet delays of the chosen codeword to the packets of the flow $f_i$.  
\begin{figure}
\begin{center}
\includegraphics[width=\textwidth/2,height=\textheight,keepaspectratio]{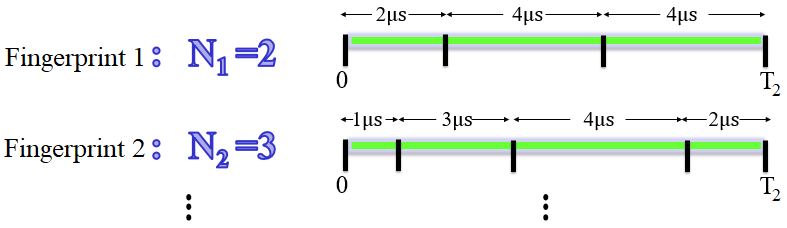}
\end{center}
 \caption{Codebook generation: Alice and Bob share a secret codebook which specifies the sequence of inter-packet delays corresponding to each fingerprint. To generate each codeword, a number $N$ is generated according to the Poisson distribution with parameter $\lambda T_2$, and then $N$ points are placed uniformly and randomly in the time interval $[0,T_2]$.}
 \label{fig:codebook}
 \end{figure}
 
\textbf{Analysis}: ({\em Invisibility}) The analysis of invisibility follows from that of covertness in~\cite[Theorem 2]{soltani2015covert}. In the first phase, Alice slows down the flows from rate $\lambda$ to rate $\lambda-\epsilon \sqrt{{2 \lambda}/{m T_1}}$, where $\epsilon>0$, while lower bounding Willie's error probability ($\mathbb{P}_{\mathrm e}^{(\mathrm w)}$) by $\frac{1}{2}-{\epsilon}$. During the second phase, the packet timings for each flow is an instantiation of a Poisson process with rate $\lambda$ and hence the traffic pattern is indistinguishable from the pattern that Willie expects to observe. Hence, the scheme is invisible.

({\em Reliability}) By~\cite[Definition 2]{verdubitsq}, Bob can successfully extract the fingerprint from $f_i$ as long as $T_2$ is large and: 
\begin{align}
\label{eq:22} \frac{\log m}{T_2}< C(q_i),
\end{align}
\noindent where $C(q_i)$ is the capacity of $q_i$ for conveying information through packet delays. By~\cite[Proposition 1]{liutiming}, $C(q_i)=\lambda \log{\left({(\mu_i-\lambda'_i)}/{\lambda}\right)}$, where $\lambda'_i$ is the sum of rates of the interfering flows passing through $q_i$. Define
\begin{align}
\label{eq:c} C=\lambda \log{\left({\underset{i}{\mathrm{min}} \{\mu_i-\lambda'_i\}}/{\lambda}\right)}.
\end{align}
\noindent Since~\eqref{eq:22} holds for all $1\leq i\leq m$, for large $T_2$: 
\begin{align}
\label{eq:cap} \frac{\log m}{T_2}<  {{\underset{i}{\mathrm{min}} \{C(q_i)\}}}=C.
\end{align}

Note that Alice does not run out of fingerprints since the number of fingerprints in her codebook equals to the number of flows. Finally, similar to the reliability analysis in~\cite{soltani2015covert}, we can show that if
\begin{align}
\label{eq:lph10} T_1&=\frac{T m\alpha/\epsilon^2}{1+m\alpha/\epsilon^2},\\
\label{eq:lph20}  T_2&=T-T_1=\frac{T}{1+m\alpha/\epsilon^2},
\end{align}
\noindent where 
\begin{align}
\label{eq:alpha}\alpha=(2 \mathrm{erf}^{-1}(1-{\zeta}))^2,
\end{align}
\noindent then $\mathbb{P}_{\mathrm{f}}\leq \zeta$.  Thus Alice's fingerprinting is reliable. 

({\em Number of flows}) By~\eqref{eq:cap} and~\eqref{eq:lph20}, we require
\begin{align}
\label{eq:29}\log m < C T_2 = \frac{C T}{1 + m \alpha \epsilon^2}.
\end{align}
\noindent Next, we show that if 
\begin{align}
\label{eq:11} m= \frac{1}{2} \min\left\{ \frac{\epsilon^2}{\alpha}\left({\frac{TC}{W(TC)}-1}\right),\frac{TC}{W(TC)}\right\},
\end{align}
\noindent then~\eqref{eq:29} is satisfied. Consider the following fact:

\begin{fact} \label{f:1} For $x,y>0$, if $x \log x = y$, then $x=y/W(y)$, where $W(\cdot)$ is the lambert-W function (see Definition~\ref{def:3}).   
\end{fact}
\noindent \textbf{Proof.} By Definition~\ref{def:3}, $W(y) e^{W(y)} = y$. Therefore, $ W(y) = \log \frac{y}{W(y)}$. Consequently, 
\begin{align}
\nonumber x \log x= \frac{y}{W(y)} \log{\frac{y}{W(y)}} = \frac{y}{W(y)} W(y) = y 
\end{align}
\QEDB 

\noindent If $m\geq 1+m \alpha/\epsilon^2$, since $m < \frac{TC}{W(TC)}$, then Fact~\ref{f:1} yields:
$$\nonumber TC > m \log{m} \geq (1+m \alpha/\epsilon^2)\log(m).$$ 

\noindent If $m< 1+m \alpha/\epsilon^2$, since $ m<\frac{\epsilon^2}{\alpha}\left({\frac{TC}{W(TC)}-1}\right)$, then:
\begin{align}
\nonumber TC &> (1+m \alpha/\epsilon^2) \log{(1+m \alpha/\epsilon^2)}= (1+m \alpha/\epsilon^2) \left(\log{m}+\log{\frac{1+m \alpha/\epsilon^2}{m}}\right) \geq (1+m \alpha/\epsilon^2) \log{m}.
\end{align}
\noindent Consequently, Alice and Bob can invisibly and reliably track $m=\mathcal{O}(T/W(T))$ flows. Note that by~\eqref{eq:lph10}, $T_1 \to \infty$ as $T \to \infty$, as required by the proof for invisibility of the second phase.  Also, by~\eqref{eq:11} and~\eqref{eq:lph20} we can show that $T_2 \to \infty$ as $T \to \infty$, as required by the proof for reliability.


\end{proof}
 
\section{Scenario~2: Each flow is fingerprinted independently with probability $p$} \label{scen2}
In this section we consider Scenario~2. In a set containing $m$ network input flows with equal rates, Alice embeds fingerprints into each flow independently with probability $p$ in the time interval $[0,T]$, and Bob extracts the fingerprints from the output flows to infer the connection between input and output flows of the network. Similar to Scenario~1, we show that employing a two phase scheme, Alice can embed a unique flow identifier fingerprint in the chosen flows by altering their packet timings according to a secret fingerprint codebook shared between Alice and Bob but unknown to Willie. We calculate the asymptotic expression for the number of flows that can be fingerprinted as a function of $T$ using this strategy
\begin{thm} \label{thm:pflows} 
 Consider the setting in Section~\ref{sec:sysmod}. In a set $\mathcal{F}$ containing $m$ flows with rates $\lambda_i=\lambda$, if Alice embeds fingerprints in each flow independently with probability $p$, Alice and Bob can invisibly and reliably track $\mathcal{O}\left({e^{ CT-\sqrt{C T \alpha}}}\right)$ flows in the time interval $[0,T]$, where $C$ and $\alpha$ are given in~\eqref{eq:c} and~\eqref{eq:alpha}, respectively, as long as 
\begin{align}
\label{eq:m1} m&= \frac{e^{ 2CT-\sqrt{C T \alpha}} }{2 \epsilon^{2}},  \\
\label{eq:p1} p&= \epsilon^2 e^{{-C  T}}.
\end{align}
\end{thm}

\begin{proof}
\textbf{Construction}: The construction is similar to that of Scenario~1. Alice's codebook contains $M$ fingerprints where
\begin{align}
\label{eq:bigM} M&=(1/2) e^{  \frac{C T}{1+ {\alpha}/{\ln\left(1+e^{{\sqrt{CT  \alpha}}} \right)}}},
\end{align}
To decide whether to embed a fingerprint in a flow or not, Alice generates independent Bernoulli random variables $X_1,\ldots,X_m$ with $\mathbb{P}(X_i=1)=p$, and she embeds a fingerprint in $f_i$ if and only if $X_i=1$.

\textbf{Analysis}: ({\em Invisibility}) We analyze the invisibility of the first and second phases separately. In the first phase, the joint pdfs of Willie's observations  under $H_0$ (Alice did not embed fingerprints), and $H_1$ (Alice embedded fingerprints) are:
\begin{align}
\nonumber \mathbb{P}_0&= \prod_{i=1}^{m} \mathbb{P}_{\lambda}(n_i),\\ 
\nonumber \mathbb{P}_1&=\prod_{i=1}^{m} \left(p\mathbb{P}_{\lambda-\Delta}(n_i)+(1-p)\mathbb{P}_{\lambda}(n_i)\right),
\end{align}
\noindent When Willie applies the optimal hypothesis test to minimize $\mathbb{P}_{\mathrm e}^{(\mathrm w)}$ \cite[Eq.1]{soltani2017covert}:
\begin{align} 
\label{eq:0} \mathbb{P}_{\mathrm e}^{(\mathrm w)} \geq \frac{1}{2}- \sqrt{\frac{1}{8} \mathcal{D}(\mathbb{P}_1 || \mathbb{P}_0)},
\end{align}
\noindent where $\mathcal{D}(\mathbb{P}_1 || \mathbb{P}_0)$ is the relative entropy between  $\mathbb{P}_1$ and  $\mathbb{P}_0$. Denote by $\mathbb{E}_{p}[\cdot]$ the expected value with respect to the probability measure $\left(p\mathbb{P}_{\lambda-\Delta}(n_i)+(1-p)\mathbb{P}_{\lambda}(n_i)\right)$. Then:
\begin{align}
\nonumber \mathcal{D}(\mathbb{P}_1||\mathbb{P}_0)&= \sum_{i=1}^{m} \mathcal{D}(p\mathbb{P}_{\lambda-\Delta}(n_i)+(1-p)\mathbb{P}_{\lambda}(n_i)||\mathbb{P}_{\lambda}(n))\\
\nonumber &= m   \mathbb{E}_{p} \left[\ln\left(p e^{\Delta T_1} \left(\frac{ \lambda -\Delta  }{\lambda  }\right)^n+(1-p)\right)\right],\\
  \nonumber & \stackrel{(a)}{\leq}   m \mathbb{E}_{p} \left[p e^{\Delta T_1} \left(\frac{ \lambda -\Delta  }{\lambda  }\right)^n-p\right],\\
& \label{eq:13}\stackrel{(b)}{=} m p^2    ( e^{ \Delta^2 T_1/\lambda}-1).
\end{align}
\noindent where $(a)$ is true since $\ln(1+x)\leq x$ for all $x \in \mathbb{R}$, and $(b)$ is true since $ \mathbb{E}_{p} \left[  \left(\frac{ \lambda -\Delta  }{\lambda  }\right)^k\right]
 =p e^{-\Delta T_1} e^{ \Delta^2 T_1/\lambda} + (1-p) e^{- \Delta T_1}$. Let  $\Delta = \sqrt{\frac{\lambda}{T_1} \ln(1+\frac{\epsilon^2}{2 m p^2})}$. Then,~\eqref{eq:13} yields $\mathcal{D}(\mathbb{P}_1||\mathbb{P}_0) \leq \epsilon^2/2$. By~\eqref{eq:0}, Willies' probability of error ($\mathbb{P}_{\mathrm e}^{(\mathrm w)}$) is lower bounded by $\frac{1}{2}-{\epsilon}$, and thus the first phase is invisible. The analysis of the invisibility for the second phase is the same as that of Scenario~1. Thus, the fingerprinting scheme is invisible. 

({\em Reliability}) Similar to the reliability analysis of Theorem~\ref{thm:allflows}, we can show that the probability that Alice runs out of packets for each flow is upper bounded by $\zeta$ as long as
\begin{align}
\label{eq:lph1} T_1&=\frac{T \alpha}{\ln(1+\frac{\epsilon^2}{2 m p^2})+\alpha},\\
\label{eq:lph2} T_2&=T-T_1=\frac{T }{1+\alpha/\ln(1+\frac{\epsilon^2}{2 m p^2})},
\end{align}
\noindent By~\eqref{eq:m1},~\eqref{eq:p1}, $\frac{\epsilon^2}{2 mp^2}= {e^{\sqrt{C T \alpha}}}$, and thus $T_1,T_2\to \infty$ as $T \to \infty$.

Next, we show that Bob can successfully extract Alice's fingerprints. By~\eqref{eq:bigM} and~\eqref{eq:lph2},
\begin{align}
\nonumber \frac{\log{M} }{T_2}&<  C
\frac{1+\alpha/\ln(1+\frac{\epsilon^2}{2 m p^2})}{ 1+ {\alpha}/{\ln\left(1+e^{{\sqrt{CT \alpha}}} \right)}}=C. 
\end{align}
\noindent where the last step is true since $\frac{\epsilon^2}{2 mp^2}= {e^{\sqrt{C T \alpha}}}$. Hence,~\eqref{eq:cap} is satisfied, and thus Bob successfully extracts each fingerprint.

Furthermore, since $M-mp$ is an increasing function of $T$, by the weak law of large numbers (WLLN) we can show that Alice does not run out of fingerprints. Hence, Alice's fingerprinting is reliable. 
  
({\em Number of flows}) By~\eqref{eq:m1},~\eqref{eq:p1}, $mp=(1/2){e^{ CT-\sqrt{C T \alpha}} }$, and by the WLLN, 
\begin{align}
\forall \gamma>0: \lim\limits_{T \to \infty} \mathbb{P}(N_f-mp>\gamma) = 0.
\end{align}
Therefore, Alice and Bob can invisibly and reliably track $N_f=\mathcal{O}(e^{ CT-\sqrt{C T \alpha}})$ flows.   
\end{proof}
\section{Discussion} \label{disc}
\subsection{Source of the gain in Scenario 2}
The result for Scenario 2 indicates a much larger fingerprint dictionary can be generated and employed covertly than in Scenario 1.  Note that~\eqref{eq:p1} implies that in Scenario 2, a small portion of the flows are fingerprinted. Intuitively, because Willie has to investigate a large number of flows to look for alterations in the timings of a relatively (very) small random subset of those flows, in particular in the first phase, this makes covertness much easier to achieve and leads to the significant gains observed.
\subsection{Extension to distinct rates}\label{scen2.5}
When Scenarios 1 and 2 are extended to distinct flow rates, Alice can build a codebook in which the rate of the codewords is $\lambda_{\mathrm{min}}=\min{(\lambda_1,\ldots,\lambda_m)}$. To embed a fingerprint in a flow $f_i$, Alice first scales the corresponding codeword $(\tau_1,\ldots,\tau_N)$ by a factor  ${\lambda_i}/{\lambda_{\mathrm{min}}}$ and applies the inter-packet delays $(\frac{\lambda_i \tau_1}{\lambda_{\mathrm{min}}} ,\ldots,\frac{\lambda_i \tau_N}{\lambda_{\mathrm{min}}})$ to the first $N+1$ packets of the flow. If Alice receives more than $N+1$ packets in the fingerprinting phase, she releases the excess packets according to random independent inter-packet delays generated from the pdf of an exponential random variable with mean $\lambda_i^{-1}$. Bob rescales the flow $f_i$ by a factor of ${\lambda_{{\mathrm{min}}}}/{\lambda_i}$ and uses the codebook to extract the corresponding fingerprint. 

We can show that if $\Delta_i = \epsilon \sqrt{{2 \lambda_i}/{m T_1}}$ and $\Delta_i = \sqrt{(\lambda_i/T_1) \ln(1+\frac{\epsilon^2}{2 m p^2})}$ in the first phases of Scenarios 1 and 2, respectively, then Alice's buffering is invisible. Note that the fingerprinted flow $f_i$ in the second phase is a realization of a Poisson process with rate $\lambda_i$, and thus it is indistinguishable from the pattern that Willie expects to observe. Hence, the scheme is invisible.

Note that the time to transmit a fingerprint in $f_i$ is $T_2  {\lambda_{\mathrm{min}}}/{\lambda_i}$. Therefore Bob can successfully extract Alice's codeword from $f_i$ as long $T_2$ is large and 
 \begin{align}
\label{eq:23} \frac{\log M}{T_2 \lambda_{{\mathrm{min}}}/\lambda_i}< \lambda_i C'(q_i)=\log{\left(\frac{\mu_i-\lambda'_i}{\lambda_i}\right)}.
\end{align}
Since~\eqref{eq:23} is true for all $1\leq i \leq m$, 
 \begin{align}
\nonumber \frac{\log M}{T_2}\leq   {{\underset{i}{\mathrm{min}} \left\{\frac{\lambda_{\mathrm{min}}}{\lambda_i}C'(q_i)\right\}}} =C'.
 \end{align}
  Finally, we can obtain an expression for the number of flows by replacing $C$ with $C'$ in the results of Theorems~\ref{thm:allflows} and~\ref{thm:pflows}.

\subsection{Extension of the network}\label{scen4}
By~\cite{mimcilovic2006mismatch}, we can extend our model to $m$ parallel routes where each route consists of multiple $M/M/1$ queues in tandem. On each route $r_i$, queues are shared between a main flow $f_i$ and interfering flows and the interfering flows are independent. Furthermore, by~\cite[Corollary 3.3]{kelly2011reversibility}, we can relax the condition of independent interference for queues on each route and extend our model to a feedforward multiclass product form network~\cite{baskett1975open}
 containing $m$ parallel routes where each route $r_i$ conveys flow $f_i$ and consists of multiple $M/M/1$ queues in tandem shared between a main flow $f_i$ and interfering flows.
\section{Conclusion} \label{sec:con}
In this paper, we presented the construction and analysis for embedding fingerprints in packet timings of flows. In a setting where a set of flows visit Alice, adversary Willie, a network of $m$ independent parallel $M/M/1$  queues with background traffic, and Bob respectively, we established a construction where Alice alters the packet timings in the time interval $[0,T]$, according to a secret codebook shared with Bob, to embed flow identifier fingerprints in them without being detected by Willie. 
We considered two scenarios: 1) Alice embeds fingerprints in all of the flows; 2) Alice embeds fingerprint in each flow independently with probability $p$, and calculated the asymptotic expression for the number of flows that can be fingerprinted as a function of $T$. 





\renewcommand{\baselinestretch}{1.02} 
\bibliographystyle{ieeetr}

\begin{thebibliography}{10}

\bibitem{staniford1995holding}
S.~Staniford-Chen and L.~T. Heberlein, ``Holding intruders accountable on the
  internet,'' in {\em Security and Privacy, 1995. Proceedings., 1995 IEEE
  Symposium on}, pp.~39--49, IEEE, 1995.

\bibitem{zhang2000detecting}
Y.~Zhang and V.~Paxson, ``Detecting stepping stones,'' in {\em USENIX Security
  Symposium}, vol.~171, p.~184, 2000.

\bibitem{syverson2001towards}
P.~Syverson, G.~Tsudik, M.~Reed, and C.~Landwehr, ``Towards an analysis of
  onion routing security,'' in {\em Designing Privacy Enhancing Technologies},
  pp.~96--114, Springer, 2001.

\bibitem{houmansadr2012design}
A.~Houmansadr, {\em Design, analysis, and implementation of effective network
  flow watermarking schemes}.
\newblock PhD thesis, University of Illinois at Urbana-Champaign, 2012.

\bibitem{wang2003robust}
X.~Wang and D.~S. Reeves, ``Robust correlation of encrypted attack traffic
  through stepping stones by manipulation of interpacket delays,'' in {\em
  Proceedings of the 10th ACM conference on Computer and communications
  security}, pp.~20--29, ACM, 2003.

\bibitem{wang2005tracking}
X.~Wang, S.~Chen, and S.~Jajodia, ``Tracking anonymous peer-to-peer voip calls
  on the internet,'' in {\em Proceedings of the 12th ACM conference on Computer
  and communications security}, pp.~81--91, ACM, 2005.

\bibitem{houmansadr2009rainbow}
A.~Houmansadr, N.~Kiyavash, and N.~Borisov, ``Rainbow: A robust and invisible
  non-blind watermark for network flows,'' in {\em NDSS}, 2009.

\bibitem{houmansadr2011swirl}
A.~Houmansadr and N.~Borisov, ``Swirl: A scalable watermark to detect
  correlated network flows,'' in {\em NDSS}, 2011.

\bibitem{houmansadr2013need}
A.~Houmansadr and N.~Borisov, ``The need for flow fingerprints to link
  correlated network flows,'' in {\em International Symposium on Privacy
  Enhancing Technologies Symposium}, pp.~205--224, Springer, 2013.

\bibitem{rezaei2017tagit}
F.~Rezaei and A.~Houmansadr, ``Tagit: Tagging network flows using blind
  fingerprints,'' {\em Proceedings on Privacy Enhancing Technologies},
  vol.~2017, no.~4, pp.~290--307, 2017.

\bibitem{bash_jsac2013}
B.~Bash, D.~Goeckel, and D.~Towsley, ``Limits of reliable communication with
  low probability of detection on {AWGN} channels,'' {\em Selected Areas in
  Communications, IEEE Journal on}, vol.~31, pp.~1921--1930, September 2013.

\bibitem{soltani2014covert}
R.~Soltani, B.~Bash, D.~Goeckel, S.~Guha, and D.~Towsley, ``Covert single-hop
  communication in a wireless network with distributed artificial noise
  generation,'' in {\em Communication, Control, and Computing (Allerton), 2014
  52nd Annual Allerton Conference on}, pp.~1078--1085, IEEE, 2014.

\bibitem{soltani2015covert}
R.~Soltani, D.~Goeckel, D.~Towsley, and A.~Houmansadr, ``Covert communications
  on poisson packet channels,'' in {\em 2015 53rd Annual Allerton Conference on
  Communication, Control, and Computing (Allerton)}, pp.~1046--1052, IEEE,
  2015.

\bibitem{soltani2016allerton}
R.~Soltani, D.~Goeckel, D.~Towsley, and A.~Houmansadr, ``Covert communications
  on renewal packet channels,'' in {\em 2016 54th Annual Allerton Conference on
  Communication, Control, and Computing (Allerton)}, IEEE, 2016.

\bibitem{soltani2017covert}
R.~Soltani, D.~Goeckel, D.~Towsley, B.~Bash, and S.~Guha, ``Covert wireless
  communication with artificial noise generation,'' {\em arXiv preprint
  arXiv:1709.07096}, 2017.

\bibitem{cormen2009introduction}
T.~H. Cormen, {\em Introduction to algorithms}.
\newblock MIT press, 2009.

\bibitem{verdubitsq}
V.~Anantharam and S.~Verdu, ``Bits through queues,'' {\em Information Theory,
  IEEE Transactions on}, vol.~42, no.~1, pp.~4--18, 1996.

\bibitem{liutiming}
X.~Liu and R.~Srikant, ``The timing capacity of single-server queues with
  multiple flows,'' {\em DIMACS Series in Discrete Mathematics and Theoretical
  Computer Science}, 2004.

\bibitem{mimcilovic2006mismatch}
P.~Mimcilovic, ``Mismatch decoding of a compound timing channel,'' in {\em
  Forty-Fourth Annual Allerton Conference on Communication, Control, and
  Computing}, 2006.

\bibitem{kelly2011reversibility}
F.~P. Kelly, {\em Reversibility and stochastic networks}.
\newblock Cambridge University Press, 2011.

\bibitem{baskett1975open}
F.~Baskett, K.~M. Chandy, R.~R. Muntz, and F.~G. Palacios, ``Open, closed, and
  mixed networks of queues with different classes of customers,'' {\em Journal
  of the ACM (JACM)}, vol.~22, no.~2, pp.~248--260, 1975.

\end{thebibliography}

\end{document}